\documentclass[a4paper,UKenglish]{article}
\usepackage{microtype}
\usepackage{bussproofs}
\usepackage{stmaryrd}
\usepackage{amsmath}
\usepackage{amssymb}
\usepackage{hyperref}

\newtheorem{theorem}{Theorem}
\newtheorem{proof}{Proof}

\newcommand{\new}{\mathsf{new}}
\newcommand{\for}{\mathrm{for }}
\newcommand{\Ncal}{\mathcal{N}}
\newcommand{\interp}[1]{\llbracket #1 \rrbracket}
\newcommand{\interpp}[1]{\{\!| #1 |\!\}}
\newcommand{\from}{\leftarrow}
\newcommand{\maps}{\colon}
\newcommand{\Th}{\mathrm{Th}}
\newcommand{\Gph}{\mathrm{Gph}}
\newcommand{\FinSet}{\mathrm{FinSet}}
\newcommand{\FPGphCat}{\mathrm{FPGphCat}}

\newcommand{\op}{\mathrm{op}}

\newcommand{\pic}{$\pi$-calculus}

\begin{document}
\title{Representing operational semantics with enriched Lawvere theories}
\author{
Michael Stay\\
  {Pyrofex Corp.}\\
  {\fontsize{8}{8}\selectfont stay@pyrofex.net}\\
  \and
  L.G. Meredith\\
  {RChain Cooperative}\\
  {\fontsize{8}{8}\selectfont greg@rchain.coop}
}
\maketitle
\begin{abstract}
  \noindent
  Many term calculi, like $\lambda$-calculus or {\pic}, involve
  binders for names, and the mathematics of bound variable names is
  subtle.  Sch\"onfinkel introduced the SKI combinator calculus in
  1924 to clarify the role of quantified variables in intuitionistic
  logic by eliminating them. Yoshida demonstrated how to
  eliminate the bound names coming from the input prefix in the
  asynchronous {\pic}, but her combinators still depend on the $\new$
  operator to bind names.  Recently, Meredith and Stay
  showed how to modify Yoshida's combinators by replacing $\new$ and
  replication with reflective operators to provide the first
  combinator calculus with no bound names into which the asynchronous
  {\pic} has a faithful embedding. Here we provide an alternative set
  of combinators built from $\mathsf{SKI}$ plus reflection that also
  eliminates all nominal phenomena, yet provides a faithful
  embedding  of a reflective higher-order pi calculus. 
  We show that with the nominal features effectively
  eliminated as syntactic sugar, multisorted Lawvere theories enriched
  over graphs suffice to capture the operational semantics of the
  calculus.
\end{abstract}

\EnableBpAbbreviations

\section{Introduction}
Many term calculi, like $\lambda$-calculus or {\pic}, involve binders
for names, and the mathematics of bound variable names is subtle.
Sch\"onfinkel introduced the SKI combinator calculus in 1924 to
clarify the role of quantified variables in intuitionistic logic by
eliminating them \cite{finkel}. Yoshida demonstrated how to eliminate
the bound names coming from the input prefix in the asynchronous
{\pic}, but her combinators still depend on the $\new$ operator to
bind names. Curry developed Sch\"onfinkel's ideas much
further. Recently, Meredith and Stay \cite{Rhocomb} showed how to
modify Yoshida's combinators by replacing $\new$ and replication with
reflective operators to provide the first combinator calculus with no
bound names into which the asynchronous {\pic} has a faithful
embedding of a reflective higher-order pi calculus. 
Here we provide an alternative set of combinators built
from $\mathsf{SKI}$ plus reflection that also eliminates all nominal
phenomena, yet provides a faithful embedding.

The recent work by Jamie Gabbay and Andrew Pitts
\cite{DBLP:journals/fac/GabbayP02} and others
\cite{DBLP:journals/jcss/Clouston14} on nominal set theory has put the
study of bound names and substitution on a much nicer foundation, at
the cost of an increase in the complexity of the semantic
framework. It interprets nominal phenomena in terms of atoms in
Fraenkl-Mostowski set theory. Clouston's work in particular makes
evident the additional machinery needed to interpret nominal phenomena
as Lawvere theories. On the other hand, with the nominal features
effectively eliminated as syntactic sugar, we show that multisorted
Lawvere theories enriched over graphs suffice to capture the
operational semantics of the calculus.

\section{Previous work}

There is a long history and an enormous body of work on modeling term rewriting and operational semantics with various notions of category enriched over category-like structures; we only have room here for a sampling.  L\"uth and Ghani \cite{DBLP:conf/ctcs/LuethG97} use poset-enriched categories to study the modularity of strong normalization.  One approach to nominality is the one we mentioned in the introduction; a different approach deals with nominal issues by allowing ``funtion types'' in the signature:  Seely \cite{DBLP:conf/lics/Seely87} 
suggested using 2-categories for modeling the denotational semantics of lambda calculus in Scott domains to capture the adjunction between $\beta$ reduction and $\eta$ conversion;  Hilken \cite{DBLP:journals/tcs/Hilken96}
expands Seely's work by exploring the proof theory using categorical logic; and Hirschowitz \cite{DBLP:journals/corr/Hirschowitz13}
generalizes algebraic signatures to cartesian closed 2-signatures.  A third approach is to model substitution explicitly: Stell \cite{Stell}
considered sesquicategories for term rewriting; in his system, objects are finite sets of variables, morphisms are substitutions, and 2-morphisms are roughly rewrite rules.

\section{Gph-enriched categories}
Here we review some standard definitions and results in enriched category theory; see \cite{CIS-335497}, \cite{Power99EnrichedLawvereTheories}, \cite{DBLP:journals/acs/LackR11}, and \cite{Trimble} for more details.

A {\bf directed multigraph with self loops}, hereafter {\bf graph}, consists of a set $E$ of edges, a set $V$ of vertices, two functions $s,t\maps E \to V$ picking out the source and target of each edge, and a function $a\maps V \to E$ such that $s\circ a$ and $t \circ a$ are both the identity on $V$---that is, $a$ equips each vertex in $V$ with a chosen self loop.  There are no constraints on $E, V, s,$ or $t$, so a graph may have infinitely many vertices and infinitely many edges between any pair of vertices.  A {\bf graph homomorphism} from $(E, V, s, t, a)$ to $(E', V', s', t', a')$ is a pair of functions $(\epsilon\maps E \to E', \upsilon\maps V \to V')$ such that $\upsilon\circ s = s' \circ \epsilon$ and $\upsilon\circ t = t' \circ \epsilon$.  {\bf Gph} is the category of graphs and graph homomorphisms.  Gph has finite products: the terminal graph is the graph with one vertex and one loop, while the product of two graphs $(E, V, s, t, a) \times (E', V', s', t', a')$ is $(E \times E', V \times V', s \times s', t\times t', a \times a').$

A {\bf Gph-enriched category} consists of
\begin{itemize}
  \item a set of objects;
  \item for each pair of objects $x, y,$ a graph $\hom(x,y);$
  \item for each triple of objects $x, y, z,$ a composition graph homomorphism $\circ\maps \hom(y, z) \times \hom(x, y) \to \hom(x, z);$ and
  \item for each object $x,$ a vertex of $\hom(x, x),$ the identity on $x,$
\end{itemize}
such that composition is associative, and composition and the identity obey the unit laws.  A Gph-enriched category has finite products if the underlying category does.

Any category is trivially Gph-enrichable by treating the elements of the hom sets as vertices and adjoining a self loop to each vertex.  The category Gph is nontrivially Gph-enriched: Gph is a topos, and therefore cartesian closed, and therefore enriched over itself.  Given two graph homomorphisms $F, F'\maps (E, V, s, t, a) \to (E', V', s', t', a'),$ a {\bf graph transformation} assigns to each vertex $v$ in $V$ an edge $e'$ in $E'$ such that $s'(e') = F(v)$ and $t'(e') = F'(v).$  Given any two graphs $G$ and $G',$ there is an exponential graph $G'^G$ whose vertices are graph homomorphisms between them and whose edges are graph transformations.

A {\bf Gph-enriched functor} between two Gph-enriched categories $C, D$ is a functor between the underlying categories such that the graph structure on each hom set is preserved, {\em i.e.} the functions between hom sets are graph homomorphisms between the hom graphs.

Let $S$ be a finite set, $\FinSet$ be a skeleton of the category of finite sets and functions between them, and $\FinSet/S$ be the category of functions into $S$ and commuting triangles.  A {\bf multisorted Gph-enriched Lawvere theory}, hereafter {\bf Gph-theory} is a Gph-enriched category with finite products Th equipped with a finite set $S$ of {\bf sorts} and a Gph-enriched functor $\theta\maps \FinSet^{\op}/S \to \Th$ that preserves products strictly.  Any Gph-theory has an underlying multisorted Lawvere theory given by forgetting the edges of each hom graph.

A {\bf model} of a Gph-theory Th is a Gph-enriched functor from Th to Gph that preserves products up to natural isomorphism.  A {\bf homomorphism of models} is a braided Gph-enriched natural transformation between the functors.  Let FPGphCat be the 2-category of small Gph-enriched categories with finite products, product-preserving Gph-functors, and braided Gph-natural transformations.  The forgetful functor $U\maps \FPGphCat[\Th, \Gph] \to \Gph$ that picks out the underlying graph of a model has a left adjoint that picks out the free model on a graph.

Gph-enriched categories are part of a spectrum of 2-category-like structures.  A strict 2-category is a category enriched over Cat with its usual product.  Sesquicategories are categories enriched over Cat with the ``funny'' tensor product \cite{Lack2010}; a sesquicategory can be thought of as a 2-category where the interchange law does not hold.  A Gph-enriched category can be thought of as a sesquicategory where 2-morphisms (now edges) cannot be composed.  Any strict 2-category has an underlying sesquicategory, and any sesquicategory has an underlying Gph-enriched category; these forgetful functors have left adjoints.

\section{Gph-theories as models of computation}

Lawvere theories and their generalizations are categories with infinitely many objects and morphisms, but most theories of interest are finitely generated.  A presentation of the underlying multisorted Lawvere theory of a finitely-generated Gph-theory is a signature for a term calculus, consisting of a set of sorts, a set of term constructors, and a set of equations, while the edges in the hom graphs of the theory encode the reduction relation.

Here is a presentation of the SKI combinator calculus as a Gph-theory:
\begin{itemize}
  \item one sort $T$, for terms
  \item term constructors
  \[\begin{array}{rl}
    S&:1 \to T\\
    K&:1 \to T\\
    I&:1 \to T\\
    (-\; -)&: T^2 \to T\\
  \end{array}\]
  \item structural congruence (no equations)
  \item rewrites
  \[\begin{array}{rl}
    \sigma&:(((S\; x)\; y)\; z) \Rightarrow ((x\; z)\; (y\; z))\\
    \kappa&:((K\; y)\; z) \Rightarrow y\\
    \iota&:(I\; z) \Rightarrow z\\
  \end{array}\]
\end{itemize}
where in the rewrites we have used expressions like $((K\; y)\; z)$ as shorthand for
\[ T\times T \xrightarrow{\mbox{\tiny left}^{-1}} 1\times T \times T \xrightarrow{K \times T \times T} T\times T \times T \xrightarrow{(-\;-)\times T} T\times T \xrightarrow{(-\;-)} T. \]

A model $M$ of this Gph-theory in Gph picks out a graph $M(T)$ of terms and rewrites.  It picks out three special vertices $S,K,$ and $I$ of $M(T)$; it equips $M(T)$ with a graph homomorphism from $M(T)^2$ to $M(T)$ that says for every pair of vertices $(u,v),$ there is a vertex $(u\;v)$, and similarly for edges; and it equips $M(T)$ with graph transformations asserting the existence of an edge out of a reducible expression to the term it reduces to.

That this Gph-theory captures the operational semantics of the SKI calculus is almost definitional: there is an edge between distinct vertices in the free model on the empty graph if and only if the source vertex is reducible to the target vertex in a single step.

It is straightforward to verify that Gph-theories suffice to capture the operational semantics of any calculus where every context is a reduction context.  This restriction on reduction contexts is a consequence of the fact that models map term constructors to graph homomorphisms: given a model $M$, a graph homomorphism $F\maps M(T) \to M(T)$, and an edge $e\maps t_1 \to t_2,$ there is necessarily an edge $F(e)\maps F(t_1) \to F(t_2).$

\section{Gph-theory for SKI with the weak head normal form evaluation strategy}
\label{whnf}
In modern programming languages, many contexts are {\em not} reduction contexts.  In Haskell, for instance, there are no reductions under a lambda abstraction: even if $t_1$ reduces to $t_2$ as a program, the term $\backslash x \to t_1$ does not reduce to $\backslash x \to t_2.$

Gph-theories can still capture the operational semantics of calculi with restrictions on reduction contexts by introducing term constructors that explicitly mark the reduction contexts.  For example, suppose that we want an evaluation strategy for the SKI calculus that only reduces the leftmost combinator when it has been applied to sufficiently many arguments, {\em i.e.} we want the {\em weak head normal form}; we can accomplish this by introducing a term constructor $R\maps T \to T$ that explicitly marks the reduction contexts.  We then add a structural congruence rule for propagating the context and modify the existing reduction rules to apply only to marked contexts.

\begin{itemize}
  \item one sort $T$, for terms
  \item term constructors
  \[\begin{array}{rl}
    S&:1 \to T\\
    K&:1 \to T\\
    I&:1 \to T\\
    (-\; -)&: T^2 \to T\\
    R&:T \to T\\
  \end{array}\]
  \item structural congruence
  \[\begin{array}{rl}
    R(x\; y) &= (Rx\; y)\\
  \end{array}\]
  \item rewrites
  \[\begin{array}{rl}
    \sigma&:(((RS\; x)\; y)\; z) \Rightarrow ((Rx\; z)\; (y\; z))\\
    \kappa&:((RK\; y)\; z) \Rightarrow Ry\\
    \iota&:(RI\; z) \Rightarrow Rz\\
  \end{array}\]
\end{itemize}

\begin{theorem}
  Let $t$ be a term in which $R$ does not appear.  Then $Rt$ reduces to $Rt',$ where $t'$ is the weak head normal form of $t.$
\end{theorem}

\begin{proof}
If we form the term $Rt$ where $t$ contains no uses of $R$, no reductions will ever take place in the right-hand argument of an application: the structural congruence and rewrite rules enforce that the $R$ context can only move to the left term in an application, never the right.  The result follows by induction on the number of steps to reach $t'.$
\end{proof}

\section{Explicit reduction contexts as gas}
The Ethereum \cite{wood2014ethereum} and RChain \cite{RChain} projects are building virtual machines on the blockchain.  Both use the concept of a linear resource called ``gas'' (as in gasoline) that is consumed as the virtual machine executes.  Gph-theories can capture the operational semantics of a calculus where reduction contexts are consumable, and thus play a role similar to that of gas \cite{DBLP:journals/corr/StayM15}.

\begin{itemize}
  \item one sort $T$, for terms
  \item term constructors
  \[\begin{array}{rl}
    S&:1 \to T\\
    K&:1 \to T\\
    I&:1 \to T\\
    (-\; -)&: T^2 \to T\\
    R&:T \to T\\
  \end{array}\]
  \item structural congruence
  \[\begin{array}{rl}
    R(x\; y) &= (Rx\; y)\\
  \end{array}\]
  \item rewrites
  \[\begin{array}{rl}
    \sigma&:(((RS\; x)\; y)\; z) \Rightarrow ((x\; z)\; (y\; z))\\
    \kappa&:((RK\; y)\; z) \Rightarrow y\\
    \iota&:(RI\; z) \Rightarrow z\\
  \end{array}\]
\end{itemize}

\begin{theorem}
  Let $t$ be a term in which $R$ does not appear; let $t'$ be the weak head normal form of $t$; let $m$ be the number of steps by which $Rt$ reduces to $Rt'$ in the calculus of section \ref{whnf}; and let $n\ge m$.  Then in this calculus, $R^n t$ reduces to $R^{n-m}t'$ in $m$ steps.
\end{theorem}

\begin{proof}
As before, if we form the term $Rt$ where $t$ contains no uses of $R$, no reductions will ever take place in the right-hand argument of an application.  Each application of the reduction rules reduces the number of $R$s by one, and structural equivalence preserves the number of $R$s.  The result follows by induction on the number of steps to reach $t'.$
\end{proof}

\section{Gph-theory for a pi calculus variant}
\label{rhocomb}
Gph-theories can capture the operational semantics of concurrent calculi as well as serial calculi like SKI above.

  Meredith and Radestock \cite{DBLP:journals/entcs/MeredithR05} describe a reflective higher-order variant of pi calculus we call the RHO calculus.  Rather than the usual replication and $\new$ operators, they have quoting and unquoting operators.  Quoting turns a process into a name and unquoting does the opposite; freshness of names is obtained using a type discipline.  They prove that there is a faithful embedding of the monadic asynchronous pi calculus into the RHO calculus.

\subsection{The RHO calculus}
\subsubsection{Syntax}
\[\begin{array}{rlr}
  P, Q &::= 0 & \mbox{the stopped process}\\ 
  &| \quad \for(y \from x)P & \mbox{input guarded process} \\ 
  &| \quad x!P & \mbox{output process}\\ 
  &| \quad P \;|\; Q & \mbox{parallel composition}\\
  &| \quad *x & \mbox{deference}\\ 
  &\\
  x, y &::= \&P & \mbox{quotation}\\ 
\end{array}\]

Note that in the original rho-calculus papers the notation was
somewhat different. The quotation and dereference constructions were
originally written, $\ulcorner P \urcorner$ and $\urcorner x \ulcorner$,
respectively. Here we have adopted a more programmer
friendly style employing the $\&$ and $*$ of the $\mathsf{C}$ programming
language for reference (quotation) and dereference, respectively. Input
guards which were written with a whimper $?$ in more traditional process calculi style  are now written in for-comprehension style as adopted
in languages like $\mathsf{Scala}$; e.g. $x?(y)P$ is written
here $\mathsf{for}( y \from x )P$.

\subsubsection{Free and bound names}
\[\begin{array}{rl}
FN(0) &= \emptyset \\
FN(\for(y \from x)P) &= \{x\}\cup (FN(P)\backslash \{y\}) \\
FN(x!P) &= \{x\}\cup FN(P) \\
\end{array}\quad\quad
\begin{array}{rl}
FN(P|Q) &= FN(P)\cup FN(Q) \\
FN(*x) &= \{x\}
\end{array}\]
\subsubsection{Structural congruence}
Structural (process) congruence is the smallest congruence $\equiv$ containing $\alpha$-equivalence and making $(|, 0)$ into a commutative monoid.
\subsubsection{Name equivalence}
Name equivalence is the smallest equivalence relation $\equiv_N$ on names such that 
\begin{center}
  \AXC{} \UIC{$\&*x \equiv_N x$} \DP and \AXC{$P \equiv Q$} \UIC{$\&P \equiv_N \&Q$} \DP.
\end{center}
\subsubsection{Substitution}
Syntactic substitution:
\[\begin{array}{rl}
  (0)\{\&Q/\&P\} &= 0\\
  (\for (y \from x) R)\{\&Q/\&P\} &= \for (z \from (x\{\&Q/\&P\})) (R\{z/y\}\{\&Q/\&P\})\\
  (x!R)\{\&Q/\&P\} &= (x\{\&Q/\&P\})!(R\{\&Q/\&P\})\\
  (R|S)\{\&Q/\&P\} &= (R\{\&Q/\&P\}) \;|\; (S\{\&Q/\&P\})\\
  (*x)\{\&Q/\&P\} &= \left\{\begin{array}{rl}
    *\&Q & \mbox{when } x \equiv_N \&Q\\
    *x & \mbox{otherwise,}
  \end{array}\right.
\end{array}\]
where
\[ x\{\&Q/\&P\} = \left\{\begin{array}{rl}
                  \&Q & \mbox{if } x\equiv_N \&P \\
                  x & \mbox{ otherwise}
                \end{array}\right.\]
and, in the rule for input, $z$ is chosen to be distinct from $\&P, \&Q,$ the free names in $Q,$ and all the names in $R.$

Semantic substitution, for use in $\alpha$-equivalence:
\[ (*x)\{\&Q/\&P\} = \left\{\begin{array}{rl}
  Q & \mbox{when } x \equiv_N \&Q\\
  *x & \mbox{otherwise}
\end{array}\right. \]

\subsubsection{Reduction rules}
We use $\to$ to denote single-step reduction.
\begin{center}
\AXC{$x_0 \equiv_N x_1$} 
\UIC{$\for(y \from x_1)P \;|\; x_0!Q \quad \to\quad P\{\&Q / y\}$} \DP \quad \quad
\end{center}

\begin{center}
\AXC{$P\to P'$}
\UIC{$P\;|\; Q \quad \to \quad P' \;|\; Q$} \DP
\end{center}

\begin{center}
\AXC{$P\equiv P'$} \AXC{$P' \to Q'$} \AXC{$Q' \equiv Q$}
\TIC{$P\to Q$} \DP
\end{center}
\subsection{RHO combinators}
We can define an embedding $\interp{-}$ of closed RHO calculus terms into a set of RHO combinators.  We follow Milner \cite{milner91polyadicpi} in thinking of an input-prefixed process $\for(x \from y)P$ as consisting of two parts: the first names the channel $y$ on which the process is listening, while the second describes the continuation $\lambda x.P$ in terms of an abstracted name.  The right hand side of the communication rule, in effect, applies the continuation to the name to be substituted.  Since the only bound names in the RHO calculus come from input prefixing, we can completely eliminate bound names by using abstraction elimination on the continuation.  Like the weak head normal form SKI calculus above, this combinator calculus uses a linear resource $C$ to reify reduction contexts.

A Gph-theory for the operational semantics of these combinators has:
\begin{itemize}
  \item one sort $T$, for terms
  \item term constructors
    \[\begin{array}{rl}
      C &: 1 \to T \\ 
      0 &: 1 \to T \\ 
      | &: 1 \to T \\ 
      \for &: 1 \to T \\ 
      ! &: 1 \to T \\ 
      \& &: 1 \to T \\ 
    \end{array}\quad\quad
    \begin{array}{rl}
      * &: 1 \to T \\ 
      S &: 1 \to T \\ 
      K &: 1 \to T \\ 
      I &: 1 \to T \\ 
      () &: T \times T \to T \\ 
    \end{array}\]
  \item structural congruence rules
    \[\begin{array}{rll}
      ((|\; 0)\; P) &= P & \mbox{unit law}\\
      ((|\; ((|\; P)\; Q))\; R) &= ((|\; P)\; ((|\; Q)\; R) & \mbox{associativity}\\
      ((|\; P)\; Q) &= ((|\; Q)\; P) &\mbox{commutativity}\\
    \end{array}\]
  \item reduction rules
    \[\begin{array}{ll}
      \sigma\maps (((S\; P)\; Q)\; R) \Rightarrow ((P\; R)\; (Q\; R)) & \mbox{action of }S \\ 
      \kappa\maps ((K\; P)\; Q) \Rightarrow P & \mbox{action of }K\\ 
      \iota\maps (I\; P) \Rightarrow P & \mbox{action of }I\\ 
      \xi\maps ((|\; C)\; ((|\; ((\for\; (\&\; P))\; Q))\; ((!\; (\&\; P))\; R))) \Rightarrow ((|\; C)\; (Q\; (\&\; R))) & \mbox{communication}\\
      \epsilon\maps ((|\; C)\;(*\; (\&\; P))) \Rightarrow ((|\; C)\; P) & \mbox{evaluation} \\
    \end{array}\]
\end{itemize}

\subsection{Embeddings}
We define an interpretation function $\interp{-}$ from RHO calculus terms into RHO combinators by

\[\begin{array}{rl}
  \interp{0} &= 0 \\
  \interp{\for(x \from \&P)Q} &= ((for\; (\&\; \interp{P}))\; \interp{Q}_x)\\
  \interp{\&P!Q} &= ((!\; (\&\; \interp{P}))\; \interp{Q}) \\
  \interp{P|Q} &= ((|\; \interp{P})\; \interp{Q})\\
  \interp{*\&P} &= (*\; (\&\; \interp{P}))
\end{array}\]
where $\interp{-}_x$ eliminates the free name $x:$
\[\begin{array}{rl}
  \interp{P}_x &= (K\; \interp{P}) \mbox{ where $x$ is not free in } P \\
  \interp{\for(y \from \&P)Q}_x &= ((S\; ((S\; (K for))\; ((S\; (K\; \&))\; \interp{P}_x)))\; \interp{\interp{Q}_y}_x) \\
  \interp{\&P!Q}_x &= ((S\; ((S\; (K\; !))\; ((S\; (K\; \&))\; \interp{P}_x)))\; \interp{Q}_x) \\
  \interp{P|Q}_x &= ((S\; ((S\; (K\; |))\; \interp{P}_x))\; \interp{Q}_x) \\
  \interp{*\&P} &= \left\{\begin{array}{ll}
    ((S\; (K\; *))\; I) & \mbox{when } \&P \equiv_N x\\
    ((S\; (K\; *))\; ((S (K\; \&))\; \interp{P}_x) & \mbox{otherwise.}
  \end{array}\right.
\end{array}\]

Consider the following sorting on RHO combinators:
\[\begin{array}{rl}
  C &: W\\
  0 &: W\\
  | &: W \Rightarrow W \Rightarrow W\\
  \for &: N \Rightarrow (N \Rightarrow W) \Rightarrow W\\
  ! &: N \Rightarrow W \Rightarrow W\\
\end{array}\quad\quad
\begin{array}{rl}
  \& &: W \Rightarrow N\\
  * &: N \Rightarrow W\\
  S &: \forall X,Y,Z.(Z \Rightarrow Y \Rightarrow X) \Rightarrow (Z \Rightarrow Y) \Rightarrow Z \Rightarrow X\\
  K &: \forall X,Y.X \Rightarrow Y \Rightarrow X\\
  I &: \forall X.X \Rightarrow X\\
\end{array}\]

The left- and right-hand sides of each of the structural congruence and rewrite rules have the sort $W,$ the interpretation of any RHO calculus term has the sort $W,$ and the result of eliminating an abstraction has the sort $N \Rightarrow W.$

We define an interpretation function $\interpp{-}$ from $W$-sorted RHO combinators not containing $C$ into the RHO calculus by
\[\begin{array}{rl}
  \interpp{0} &= 0\\
  \interpp{P\;|\;Q} &= \interpp{P} \;|\; \interpp{Q}\\
  \interpp{((\for\; (\&\; P))\; Q)} &= for (\&\interpp{R} \leftarrow \&\interpp{P})\interpp{(Q\; (\&\; R))}\\
  \interpp{((!\; (\&\; P))\; Q)} &= \&\interpp{P}!\interpp{Q}\\
  \interpp{(*(\&\; P))} &= *\&\interpp{P}\\
  \interpp{(((S\; P)\; Q)\; R)} &= \interpp{((P\; R)\; (Q\; R))}\\
  \interpp{((K\; P)\; Q)} &= \interpp{P}\\
  \interpp{(I\; P)} &= \interpp{P}
\end{array}\]
where $R$ is any $W$-sorted RHO combinator.

Some simple calculation shows that
\begin{theorem}
\label{roundtrip}
  $\interpp{\interp{P}}$ is $\alpha$-equivalent to $P$, $Q$ is reducible to $\interp{\interpp{Q}}$ without using the rewrite $\xi,$ and $\interp{\interpp{-}}$ is idempotent.
\end{theorem}
See the appendix for more details.

\subsection{Barbed bisimilarity}
An {\bf observation relation} $\downarrow_\Ncal$ over a set of names $\Ncal$ is the smallest relation satisfying
\begin{center}
  \AXC{$y \in \Ncal$} \AXC{$x \equiv_N y$} \BIC{$x!P \downarrow_\Ncal x$} \DP $\quad$ and $\quad$ \AXC{$P  \downarrow_\Ncal x \mbox{ or } Q \downarrow_\Ncal x$} \UIC{$P\;|\;Q  \downarrow_\Ncal x$} \DP
\end{center}
for the RHO calculus or
\begin{center}
  \AXC{$y \in \Ncal$} \AXC{$x \equiv_N y$} \BIC{$((!\; x)\; P) \downarrow_\Ncal x$} \DP $\quad$ and $\quad$ \AXC{$P  \downarrow_\Ncal x \mbox{ or } Q \downarrow_\Ncal x$} \UIC{$((|\; P)\; Q)  \downarrow_\Ncal x$} \DP.
\end{center}
for the RHO combinators.

We denote eventual reduction by $\to^*$ and write $P \downarrow^*_\Ncal x$ if there exists a process $Q$ such that $P \to^* Q$ and $Q \downarrow_\Ncal x.$

An {\bf $\Ncal$-barbed bisimulation} over a set of names $\Ncal$ is a symmetric binary relation $S_\Ncal$ between agents such that $P \mathop{S_\Ncal} Q$ implies
\begin{enumerate}
  \item if $P \to P'$ then $Q \to^* Q'$ and $P' \mathop{S_\Ncal} Q',$ and 
  \item if $P \downarrow_\Ncal x,$ then $Q \downarrow^*_\Ncal x.$
\end{enumerate}
$P$ is $\Ncal$-barbed bisimilar to $Q,$ written $P \approx Q,$ if $P \mathop{S_\Ncal} Q$ for some $\Ncal$-barbed bisimulation $S_\Ncal.$

\subsection{Faithfulness}

\begin{theorem}
  $P \approx_{\mbox{\tiny calc}} Q \iff  ((|\; C)\; \interp{P}) \approx_{\mbox{\tiny comb}} ((|\; C)\; \interp{Q})$.
\end{theorem}

\begin{proof}[Proof sketch]
The only occurrence of $C$ on the right is at the topmost context and the rewrite rules preserve the location of $C$, so the only reduction context is the topmost one.  The rest follows from the two interpretation functions and theorem \ref{roundtrip}.  In particular, while the only reduction rule in the RHO calculus is synchronizing on a name, there are extra reduction rules for the RHO combinators; however, these extra reduction rules never send or receive on a name and never prevent sending or receiving on a name.  Therefore, each synchronization in the evaluation of a RHO calculus term corresponds to a synchronization in the corresponding RHO combinator term and some number of reductions of $S,K,I,$ or evaluating a quoted process.
\end{proof}

In fact, we believe a much stronger property than bisimilarity should hold: since $S,K,$ and $I$ are only used for eliminating dummy variables and the $\epsilon$ reduction plays the role of semantic substitution, $\interp{\interpp{-}}$ should pick out a normal form for a RHO combinator.  We should get a set of normal-form equivalence classes of $W$-sorted RHO combinators that is isomorphic to the set of $\alpha$-equivalence classes of RHO calculus terms.  Then we should get
\[ P \xrightarrow{\mbox{\tiny comm}} P' \quad \iff \quad \interp{P} \xrightarrow{\xi} \interp{P'}\]
and
\[ Q \xrightarrow{\xi} Q' \quad \iff \quad \interpp{Q} \xrightarrow{\mbox{\tiny comm}} \interpp{Q'}, \]
where we now regard the left and right sides as being equivalence classes.

\section{Conclusion and future work}
This paper is part of a pair of papers demonstrating that reflection
provides a powerful technique for treating nominal phenomena as
syntactic sugar, thus paving the way for simpler semantic treatments
of richly featured calculi, such as the {\pic} and other calculi of
concurrency. We illustrated the point by providing faithful semantics
of both the $\lambda$-calculus and the {\pic} in terms of graph-enriched
Lawvere theories. This work may be considered preparatory for a more
elaborate study of logics for concurrency in which the nominal
phenomena have logical status, but may be treated in a technically
simpler fashion.

\section{Appendix: abstraction elimination calculations}
\[\begin{array}{rl}
  & ((K\; \interp{P})\; x) \\
  = & \interp{P}\\
  \\
  
  & (((S\; ((S\; (K\; for))\; ((S\; (K\; \&))\; \; \interp{P}_x)))\; \interp{\interp{Q}_y}_x)\; x)\\
  =&  ((((S\; (K\; for))\; ((S\; (K\; \&))\; \; \interp{P}_x))\; x)\; (\interp{\interp{Q}_y}_x\; x))\\
  =&  ((((K\; for)\; x)\; (((S\; (K\; \&))\; \; \interp{P}_x)\; x))\; \interp{Q}_y)\\
  =&  ((for\; (((K\; \&)\; x)\; (\interp{P}_x\; x)))\; \interp{Q}_y)\\
  =&  ((for\; (\&\; \interp{P}))\; \interp{Q}_y)\\
  \\
  & (((S\; ((S\; (K\; !))\; ((S\; (K\; \&))\; \interp{P}_x)))\; \interp{Q}_x)\; x)\\
  =&  (((((S\; (K\; !))\; ((S\; (K\; \&))\; \interp{P}_x)))\; x)\; (\interp{Q}_x\; x))\\
  =&  ((((K\; !)\; x)\; (((S\; (K\; \&))\; \interp{P}_x)\; x))\; \interp{Q})\\
  =&  ((!\; (((K\; \&)\; x)\; (\interp{P}_x\; x)))\; \interp{Q})\\
  =&  ((!\; (\&\; \interp{P}))\; \interp{Q})\\
  \\
  & (((S\; ((S\; (K\; |))\; \interp{P}_x))\; \interp{Q}_x)\; x)\\
  =& ((((S\; (K\; |))\; \interp{P}_x)\; x)\; (\interp{Q}_x\; x))\\
  =& ((((K\; |)\; x)\; (\interp{P}_x\; x))\; \interp{Q})\\
  =& (|\; \interp{P})\; \interp{Q})\\
  \\
  & (((S\; (K\; *))\; I)\; x)\\
  =&  (((K\; *)\; x)\; (I\; x))\\
  =&  (*\; x)\\
  \\
  & (((S\; (K\; *))\; ((S\; (K\; \&))\; \interp{P}_x))\; x)\\
  =&  (((K\; *)\; x)\; (((S\; (K\; \&))\; \interp{P}_x)\; x))\\
  =&  (*\; (((K\; \&)\; x)\; (\interp{P}_x\; x)))\\
  =&  (*\; (\&\; \interp{P}))\\
\end{array}\]

\bibliographystyle{amsplain}
\bibliography{calco}
\end{document}